\documentclass[sn-mathphys]{sn-jnl}

\usepackage{amsmath,amsfonts,amssymb,amsthm}
\usepackage{enumitem}[shortlabels]
    \setitemize{noitemsep,topsep=0pt,parsep=0pt,partopsep=0pt}  
    \setenumerate{noitemsep,topsep=0pt,parsep=0pt,partopsep=0pt}  
\usepackage{wrapfig}
\usepackage{float}     
\usepackage{thm-restate}
\usepackage{hyperref}       
    \usepackage{url}            
\usepackage{cleveref}
\usepackage{algorithm}
    \algtext*{EndFor}       
    \algtext*{EndIf}        
\usepackage{float} 
\usepackage{adjustbox}
\usepackage{tablefootnote}

\theoremstyle{thmstyleone}%

\theoremstyle{thmstyletwo}%
\newtheorem{example}{Example}%

\theoremstyle{thmstylethree}%
\newtheorem{definition}{Definition}%

\providecommand{\corollaryname}{Corollary}
\providecommand{\claimname}{Claim}

\providecommand{\lemmaname}{Lemma}
\providecommand{\notationname}{Notation}

\providecommand{\problemname}{Problem}

\providecommand{\conjecturename}{Conjecture}
\providecommand{\experimentname}{Experiment}

\DeclareMathAlphabet{\mathpzc}{OT1}{pzc}{m}{it}
\DeclareMathSymbol{\shortminus}{\mathbin}{AMSa}{"39}

\newboolean{commentsactivated}
\setboolean{commentsactivated}{true}
\newcommand{\dm}[1]{\ifthenelse{\boolean{commentsactivated}}{{\color{blue} {\em DM: #1 }}}{}}
\newcommand{\vl}[1]{\ifthenelse{\boolean{commentsactivated}}{{\color{purple} {\em VL: #1 }}}{}}
\newcommand{\vk}[1]{\ifthenelse{\boolean{commentsactivated}}{{\color{red} {\em VK: #1 }}}{}}
\newcommand{\dhl}[1]{\ifthenelse{\boolean{commentsactivated}}{{\color{teal} {\em DH: #1 }}}{}}

\newcommand{\eg}{{for example,}~}
\newcommand{\ie}{{that is,}~}

\newcommand{\defword}[1]{\textbf{\boldmath{#1}}}

\newcommand{\N}{\mathbb{N}}
    \newcommand{\Reals}{\mathbb{R}}
\newcommand{\mc}{\mathcal}
\newcommand{\mb}{\mathbb}

\newcommand{\actions}{\mc A}
\newcommand{\action}{a}
\newcommand{\actionAlt}{a'}
\newcommand{\policy}{\pi}
\newcommand{\policies}{\Pi}
\newcommand{\chance}{\textnormal{c}}

\newcommand{\pl}{i}

\newcommand{\histories}{\mc H}
\newcommand{\history}{h}

\newcommand{\suffix}{\sqsupset}                   
\newcommand{\gameroot}{\textnormal{root}}

\newcommand{\gameRoot}{\textnormal{root}}
\newcommand{\leaf}{z}
\newcommand{\leaves}{\mc Z}

\newcommand{\infostate}{s}
\newcommand{\infostateAlt}{t}
\newcommand{\infostates}{\mc S}

\newcommand{\ims}{\textnormal{ims}}     
\newcommand{\observation}{o}            

\newcommand{\privState}{\infostate'}
\newcommand{\privStateAlt}{\infostateAlt'}
\newcommand{\privStates}{\infostates'}

\newcommand{\public}{\textnormal{pub}}
\newcommand{\publicTree}{\mc S_\public}
\newcommand{\publicState}{s_\public}

\newcommand{\publicStateAltAlt}{t_\public}
\newcommand{\publicTerminalState}{\leaf_\public}

\newcommand{\reachProb}{P}

\newcommand{\cf}[1]{{#1, \textnormal{cf}}}
\newcommand{\cfSubscript}{\textnormal{cf}}
\newcommand{\historyValue}{v}
\newcommand{\infostateValue}{V}
\newcommand{\historyActionValue}{q}
\newcommand{\infostateActionValue}{Q}

\newcommand{\infosetValue}{\infostateValue}
\newcommand{\infosetActionValue}{\infostateActionValue}

\newcommand{\BGAbbrev}{SBG}
\newcommand{\BGFull}{Sequential Bayesian game}

\newcommand{\fosgAbr}{FOSG}

\newcommand{\vanillaCFR}{\texttt{Vanilla-CFR}}
\newcommand{\publicStateCFR}{\texttt{PS-CFR}}
\newcommand{\openSpielCFR}{\vanillaCFR{}\ \texttt{(mem-eff)}}
\newcommand{\saveCFR}{\vanillaCFR{}}
\newcommand{\genericPSCFR}{\publicStateCFR{}}
\newcommand{\pokerPSCFR}{\publicStateCFR{}\ \texttt{(dom-sp)}}

\newcommand{\regret}{R}
\newcommand{\reachProbCode}{\reachProb}

\newcommand{\regUpdate}{\texttt{RegretUpdate}}
\newcommand{\psRegUpdate}{\texttt{PS-RegretUpdate}}
\newcommand{\histRegUpdate}{\texttt{Hist-RegretUpdate}}

\newcommand{\CFR}{\texttt{CFR}}
\newcommand{\regretMatching}{\texttt{RM}}
\newcommand{\chanceWeightedUtilities}{\texttt{ChWU}}

\newcommand{\bigO}{O}
\newcommand{\increment}{\mathrel{\raisebox{0.19ex}{$\scriptstyle+$}}=}

\newcommand{\supp}{\mathpzc{supp}}                
\newcommand{\utility}{u}

\newcommand{\opp}{j}          
\newcommand{\others}{{\textnormal{-}\pl}}

\newcommand{\modelAbbrev}{FOSG}
\newcommand{\modelName}{factored-observation stochastic game}

\newcommand{\initState}{w^{\textnormal{init}}}

\newcommand{\priv}[1]{{\textnormal{priv(#1)}}}    

\newcommand{\game}{{\mc{G}}}

\jyear{2022}%

\begin{document}

\title[Revisiting Game Representations: Hidden Costs of Efficiency]{Revisiting Game Representations: The~Hidden Costs of Efficiency in~Sequential Decision-making Algorithms}

\author*[12]{\fnm{Vojt\v{e}ch} \sur{Kova\v{r}\'ik}}\email{vojta.kovarik@gmail.com}

\author[2]{\fnm{David} \sur{Milec}}\email{milecdav@fel.cvut.cz}

\author[2]{\fnm{Michal} \sur{\v{S}ustr}}\email{michal.sustr@aic.fel.cvut.cz}

\author[2]{\\\fnm{Dominik} \sur{Seitz}}\email{dominik.seitz@aic.fel.cvut.cz}

\author*[2]{\fnm{Viliam} \sur{Lis\'y}}\email{viliam.lisy@agents.fel.cvut.cz}

\affil[1]{\orgdiv{Foundations of Cooperative AI Lab, CS department},\\ \orgname{Carnegie Mellon University}, \orgaddress{\street{5000 Forbes Avenue}, \city{Pittsburgh},\\ \postcode{PA 15 213}, \country{United States}}}

\affil[2]{\orgdiv{Artificial Intelligence Center, Faculty of Electrical Engineering}, \orgname{Czech Technical University in Prague}, \orgaddress{\street{Technicka 2},\\ \city{Prague}, \postcode{166 27}, \country{Czech Republic}}}

\abstract{
Recent advancements in algorithms for sequential decision-making under imperfect information have shown remarkable success in large games such as limit- and no-limit poker. 
These algorithms traditionally formalize the games using the extensive-form game formalism, which, as we show, while theoretically sound, is memory-inefficient and computationally intensive in practice.
To mitigate these challenges, a popular workaround involves using a specialized representation based on player specific information-state trees.
However, as we show, this alternative significantly narrows the set of games that can be represented efficiently.

In this study, we identify the set of large games on which modern algorithms have been benchmarked as being naturally represented by Sequential Bayesian Games.
We elucidate the critical differences between extensive-form game and sequential Bayesian game representations, both theoretically and empirically. 
We further argue that the impressive experimental results often cited in the literature may be skewed, as they frequently stem from testing these algorithms only on this restricted class of games.
By understanding these nuances, we aim to guide future research in developing more universally applicable and efficient algorithms for sequential decision-making under imperfect information.
}


\maketitle

\section{Introduction}\label{sec:intro}
An important part of intelligent behaviour is the ability to deal with uncertainty in multiagent interactions and can be modeled by games~\cite{nisan2007algorithmic}.
This area of research has recently seen a lot of progress by scaling up to solve larger and larger games, such as
    essentially solving two-player limit poker~\cite{tammelin2015solving},
    outperforming professional human players in no-limit poker~\cite{DeepStack,Libratus,Pluribus},
    or expert-level play in the strategic and negotiation game Diplomacy~\cite{Cicero}.

A common objective is to find a minimax equilibrium (or its approximation) in two-player zero-sum games.
Typically, when finding an equilibrium in a large game, one relies on some factorization to make the model of the game smaller, but still strategically equivalent to the original game~\cite{koller1994fast,FOG,zhang2020sparsified}.
Choosing an appropriate model is important, and can yield the difference between the algorithm terminating within a day or not even within a human lifetime.
While most literature is concerned with development of efficient algorithms, in this paper, we highlight the importance of \emph{natural models} of games and how the underlying game representations impact the practical performance of algorithms built on top of them.

We will use two examples to motivate the concept of natural modeling.
First, consider a formalization of the classical Rock-Paper-Scissors game.
While the normal-form game\footnotemark{} (NFG) model has 3 actions for each player and a utility function that can be represented by a matrix,
    the extensive-form game\footnotemark[\value{footnote}] (EFG) model additionally requires 13 histories
        (1 for first player, 3 for second player and 9 for terminals) and 2 infostates.
Not only do we need to include these additional histories and infostates, but the formalism is also longer to describe due to the introduction of these additional structures which are not needed in an NFG.
    \footnotetext{See \Cref{fig:rps_nfg_vs_efg} for description.}
Clearly, the NFG is a smaller and simpler model for Rock-Paper-Scissors and is more appropriate than an EFG. 

Second, as a contrast, consider formalization of the game Kuhn poker~\cite{kuhn1950simplified}.
Kuhn Poker can be modeled as an NFG with $64$ actions for each player.
However, if we model Kuhn poker as an EFG, we need only $12$ decision states with $2$ actions at each state.
Specifying the strategy for the NFG model requires $63$ numbers, while the EFG model only needs $12$ numbers.
As a result, the EFG model factorizes the NFG model and decreases the size of the strategy space over 5-fold.

From these examples, we observe some games have more appropriate models than others.
As a result, we refer to any model that provides the ``just right'' amount of factorization of the game, without introducing an overhead of unnecessary formalism, as one that \emph{naturally models} the game.
We refrain from giving an exact definition, as we found it is difficult to quantify formalisms based on some measure.
However, generally speaking, a natural model has a smaller game tree or strategy space than other models, which helps to build faster algorithms.

One might wonder whether there are games which are naturally modeled under formalisms other than NFGs or EFGs.
In this paper, we answer this question in the positive by observing that games like Poker, Liar's dice or Battleship can be naturally modeled as what we call Sequential Bayesian Games (\BGAbbrev{}s).
Moreover, we find that using \BGAbbrev{}s on these games as basis for these algorithms is not just intellectual curiosity, but is in fact necessary to find equilibria in these games within reasonable time and memory footprint, and is behind the mentioned successes in large games.

One of the key ingredients these successes have in common is counterfactual regret minimization (CFR),
    an algorithm for finding strong strategies for sequential decision-making under imperfect information~\cite{CFR}.
We describe, analyse, and discuss the implications of \emph{public-state CFR},
    a reformulation of the CFR algorithm that allows for exploiting the structure of \BGAbbrev{}s
        (i.e., public observations available in some games).
We further argue that the impressive experimental results cited in the literature may be skewed,
    as they frequently stem from testing CFR-like algorithms only on \BGAbbrev{}s, with CFR implemeneted as public-state CFR.
By understanding these nuances, we aim to guide future research in developing more universally applicable and efficient algorithms for sequential decision-making under imperfect information.

\begin{figure}
    \centering
    \begin{minipage}{.29\textwidth}
        \begin{tabular}{|c|c|c|c|}
            \hline
             & \textbf{r} & \textbf{p} & \textbf{s}  \\ \hline
             \textbf{R} & 0 & -1 & 1 \\ \hline
             \textbf{P} & 1 & 0 & -1 \\ \hline
             \textbf{S} & -1 & 1 & 0 \\ \hline
        \end{tabular}
    \end{minipage}
    \hfill
    \begin{minipage}{.69\textwidth}
        \includegraphics[width=\linewidth]{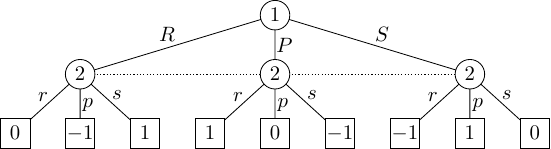}    
    \end{minipage}
    \caption{Comparison of the representation of the Rock paper scissors game. On the left is the normal-form representations and on the right is the extensive-form representation.}
    \label{fig:rps_nfg_vs_efg}
\end{figure}

\subsubsection*{Contribution}
In the first part of this paper,
    we describe public-state CFR in the context of factored-observation stochastic games (FOSGs)~\cite{FOG}.
Since FOSGs are a model which allows one to talk about public information,
    this makes it easier to decompose the game into subgames, and run the algorithm more efficiently.
Public-state CFR traverses the game tree on the level of public states  (rather than histories, as is usual for the ``vanilla'' CFR.).
We prove that the algorithm is no slower than vanilla CFR, but allows for significant speedups in some domains (\Cref{thm:PSCFR-complexity}).
We argue that the algorithm also allows for additional speedups by being more amenable to parallelisation and domain-specific implementations.
\smallskip
    
Second, we describe the class of \emph{\BGFull{}s}, where players have perfect information about everything except for the opponents' ``hidden type''.
    A canonical example of such game is hold'em poker, where each player's type is equal to their private cards.
We note that every game can be rewritten as a \BGFull{}.
    As a result, we think of games as \textit{being represented as} an \BGAbbrev{}, rather than \textit{being} an \BGAbbrev{}.
    This should not come as a surprise, since the same is true for normal- and extensive-form games.
    However, as with NFGs and EFGs, the \BGAbbrev{} representation is more natural for some games than others --- hence our claim about hold'em poker.

We argue that there is a connection between being \BGAbbrev{} and being particularly amenable to PS-CFR.
In particular, it seems noteworthy that the two most popular benchmarks in the CFR literature -- hold'em poker and (some versions of) liar's dice -- are both naturally modelled using this class of games.
\smallskip

Our final contribution is
    an empirical comparison of PS-CFR and vanilla CFR
    and discussion of the implications for CFR's performance in novel domains.
More specifically,
    we observe that PS-CFR significantly outperforms vanilla CFR on a poker subgame,
    and argue that many of the past CFR publications use an implementations that is closer to PS-CFR than to vanilla CFR (and predominantly test it on poker).
This suggests two things.
First,
    if one hopes to replicate CFR's ``advertised'' performance,
    it isn't sufficient to implement its vanilla version.
Second,
    since PS-CFR's superiority to vanilla CFR relies on the presence of public information (Prop.\,\ref{thm:Vanilla-CFR-complexity}, Thm.\,\ref{thm:PSCFR-complexity}),
    one should expect both algorithms' speed to be similar (i.e., slow) in domains where public information is scarce (e.g., blind chess).
        This does not necessarily imply that CFR is not promising for such games.
        However, it does mean that one might have to come up with a custom implementation of CFR that exploits the specific properties of the given game,
            costing precious programmer time.

\subsubsection*{Outline}
The remainder of the paper is structured as follows:
In Section~\ref{sec:background}, we formally describe factored-observation stochastic games and the vanilla version of CFR.
In Section~\ref{sec:ps-cfr},
    we describe public-state CFR for \fosgAbr{}s,
    analyze its complexity,
    and discuss its potential for practical improvements.
In \Cref{sec:beg}, we introduce sequential Bayesian games and discuss them. 
In Section~\ref{sec:empirical}, we illustrate our claims by presenting an empirical comparison of vanilla- and public-state CFR on poker.
In Section~\ref{sec:related-work}, we mention the most relevant literature.
In Section~\ref{sec:discussion}, we discuss the results and their implications.

\subsubsection*{Novelty}
The idea behind public-state CFR appears in several previous works,
    so our contribution lies not in coming up with the algorithm,
    but in describing it for a general setting,
    proving its properties,
    and running the presented experiments.
Similarly, models similar to sequential Bayesian games appear in several prior works.
    Our contribution is therefore not in inventing the model
    but in identifying it as highly relevant to CFR
    (and perhaps also in describing it in modern terms).
For more details on the related works, see \Cref{sec:related-work}.
Finally, the least tangible but nonetheless important contribution is in
    observing that more research might be needed in games where public information is less abundant than in poker.

\section{Background}\label{sec:background}

A \defword{\modelName{}} (\modelAbbrev) is a tuple $G = \left< \mc N, \mc W, p, \initState, \mc A, \mc T, \mc R, \mc O \right>$, where
$\mc N$ is the \defword{player set},
$\mc W$ is the set of \defword{world states},
$\initState$ is a designated \defword{initial state},
$p : \mc W \to 2^{\mc N}$ is a \defword{player function},
$\mc A$ is the space of \defword{joint actions},
$\mc T : \mc W \times \mc A \hookrightarrow \Delta \mc W$ is the \defword{transition function},
$\mc R : \mc W \times \mc A \hookrightarrow \Reals^{\mc N}$ is the \defword{reward function},
$\mc O : \mc W \times \mc A \times \mc W \hookrightarrow \mb O$ is the \defword{observation function},
and we have:

\medskip

\begin{itemize}
    \item $\mc N = \{1,\dots,N\}$ for some $N\in \mb N$.
    \item $\mc W$ is compact. For formal convenience, we assume that $p(\initState) = \emptyset$.
    \item $\mc A = \prod_{i \in \mc N} \mc A_i$, where each $\mc A_i$ is an arbitrary set of \defword{$i$'s actions}.
        \begin{itemize}
            \item For each $i\in p(w)$, $\mc A_i(w) \subset \mc A_i$ denotes a non-empty compact set of $i$'s (legal) \defword{actions at $w$}.
                We denote $\mc A(w) := \prod_{i \in p(w)} \mc A_i(w)$.
            \item We denote $\mc A_i(w) := \{noop\}$ for $i\notin p(w)$, which allows us to identify each $a \in \mc A(w)$ with an element of $\prod_{i \in \mc N} A_i(w)$ by appending to it the appropriate number of \textit{noop}\footnote{Recall that \textit{noop} stands for ``no operation'' ---  an action that makes no changes.} actions.
        \end{itemize}
    \item The transition probabilities $\mc T(w,a) \in \Delta \mc W$, $a\in \mc A(w)$, are defined for all $w\in \mc W$ with non-empty $p(w)$ and for some $w$ with no active players.
        \begin{itemize}
            \item A world state with $p(w)=\emptyset$ and undefined $\mc T(w,a)$ is called \defword{terminal}.
        \end{itemize}
    \item $\mc R (w,a) = (\mc R_i(w,a))_{i \in \mc N}$ for each non-terminal state $w$ and $a\in \mc A(w)$.
    \item $\mc O$ is factored into \defword{private observations} and \defword{public observations} as $\mc O = (\mc O_{\textnormal{priv}(1)},\dots,\mc O_{\textnormal{priv}(N)}, \mc O_{\textnormal{pub}} )$.
    \begin{itemize}
        \item $\mb O = \prod_{i\in\mc N} \mb O_{\textnormal{priv}(i)} \times \mb O_{\textnormal{pub}}$, where $\mb O_{(\cdot)}$ are arbitrary sets (of possible observations).
        \item We assume that $\mc O_{(\cdot)}(w,a,w') \in \mb O_{(\cdot)}$ is defined for every non-terminal $w$, $a \in \mc A(w)$, and $w'$ from the support\footnote{For finite $\mc W$, being in support of $\mc T(w,a)$ is equivalent to having a non-zero probability.} of $\mc T(w,a)$.
    \end{itemize}
\end{itemize}

\bigskip

The game starts by transitioning from the initial state $\initState$ (where no player acts) to some new state $w^1 \sim \mc T(\initState, \textnormal{\textit{noop}})$.
    This generates the initial observations that notify players that the game has started.
The game then proceeds in rounds:
    At each world state $w^k$, the active players $i \in p(w^k)$ select actions $a^k_i \in \mc A_i(w^k)$ (more on this below).
    The game then randomly transition to some $w^{k+1} \sim \mc T(w^k, \actions^k)$,
    which generates observations $o^k_i = (o^k_{\priv(i)}, o^k_\public) = \mc O_i(w^k, a^k, w^{k+1})$.
    Each transition $(w^k, a^k, w^{k+1})$ also generates rewards $r^k_i = \mc R_i(w^k, a^k, w^{k+1})$.
The game ends upon reaching some terminal world state.
    (In this paper, we assume $G$ is s.t. this always happens.)

We describe the playthroughs of the game in terms of  \defword{histories} --- i.e., sequences of the form $h = w^0 a^0 w^1 a^1 \dots w^k$
    (where $w^0 = \initState$, $a^t \in \actions(w^t)$, $w^{l+1} \in \supp(\mc T(w^l, a^l))$ ).
    We denote the set of all histories by $\histories$ and use $z$ and $\mc Z \subset \mc H$ to denote terminal histories
        (those which end in a terminal world state).
    For terminal $z \in \mc Z$, we define \defword{$i$'s utility} $u_i(z)$ as the sum of rewards $r^k_i$ received by $i$ along $z$.
    
To model the behaviour of the players,
    we assume that over the course of the game, each player $i$ maintains an \defword{information state} (infostate) $s_i$,
        which consists of the sequence of observations they have received so far, and actions they have taken.
    We sometimes factor infostates into a private infostate and public state,
        writing $\infostate_\pl = (\privState_\pl, \publicState)$.
        Each $\privState_\pl$ consists only of $i$'s private observations and actions while the latter consists only of public observations.
        We denote the corresponding sets as $\infostates_\pl$, $\privStates_\pl$, and $\publicTree$.
    We assume that $\infostate_\pl$ allows $\pl$ to infer their legal actions.\footnote{
            That is, we assume that
            $
                \forall h, h' \in \mc H :
                s_i(h) = s_i(h')
                \implies
                \actions_i(w(h)) = \actions_i(w(h'))
            $.
        }
    We thus model the players as selecting their actions using \defword{behavioural strategies} of the form
        $
            \policy_i :
                \infostates_\pl
                \to 
                \Delta(\actions_\pl)
        $, where
        $ 
                \infostate_\pl \mapsto \policy_i( \, \cdot \, \mid \infostate_\pl )
                \in \Delta (\actions_\pl(\infostate_\pl))
        $.
    We denote the corresponding sets as $\policies_\pl$ and $\policies = \prod_{\pl = 1}^N \policies_\pl$.
    We say that a \defword{policy profile} $\policy \in \policies$ is a \defword{Nash equilibrium} if $\forall \pl \forall \policy'_\pl \in \policies_\pl : u_\pl(\policy_\pl, \policy_\others) \geq u_\pl(\policy'_\pl, \policy_\others)$
        (where $\others$ corresponds to the players other than $\pl$ and $u_\pl(\policy)$ denotes $\pl$'s expected utility under $\policy$).
\subsection{Counterfactual Regret Minimization}\label{sec:vanilla-CFR}

Counterfactual regret minimization is a popular self-play algorithm for imperfect-information games \cite{CFR}.
It approximates a Nash equilibrium by iteratively traversing the game tree and minimizing a particular notion of regret --- called counterfactual regret --- at each decision point.
(Where regret measures how much better off a player could have been if they changed all their actions at the given decision point $\infostate_\pl$ to a specific $\action_\pl \in \actions_\pl(\infostate_\pl)$, and everything else remained constant.
The \textit{counterfactual} part refers to assigning weights to iterations proportional to the probability of encountering $\infostate_\pl$ at the given iteration in the counterfactual scenario where $\pl$ always selects actions that lead to $\infostate_\pl$.)
While CFR is not difficult to define formally, developing intuitions for the underlying concepts can be quite challenging and doing so would be outside of the scope of this text.
In this paper, we thus focus on giving the formal definitions while referring the reader interested in the intuitions to \cite{seitz2023value}.

\medskip

To describe CFR,
    we first need the notion of a \defword{reach probabilities} and their decomposition into the contributions of the individual players and chance.
Formally, we set
    $P_{c}(h) := \prod_{h'aw \sqsubseteq h}\mc{T}(h',a,w)$,
    $P^{\pi}_{\pl}(h) := \prod_{h'aw \sqsubseteq h}\pi_\pl(a_\pl \mid s_\pl(h'),a_\pl)$,
    and $\reachProb^\policy(\history) = \reachProb_\chance(\history) \prod_{\pl = 1}^N \reachProb^\policy_\pl(\history)$.
We also define the \defword{counterfactual reach probability} of $h$ as
    $P^{\pi}_{-i}(h) := P_{c}(h)\prod\nolimits_{j \neq \pl} P^{\pi}_j(h)$.
For infostates, we set
    $\reachProb^\policy(\infostate_\pl) := \sum_{h \in \histories, \, \infostate_\pl(h) = \infostate_\pl} \reachProb^\policy(h)$,
    $\reachProb^\policy_\pl(\infostate_\pl) := \reachProb^\policy(h_0)$ for an arbitrary $h_0 \in \histories$ with $\infostate_\pl(h_0) = \infostate_\pl$, and
    $\reachProb^\policy_\others(\infostate_\pl) := \sum_{h \in \histories, \, \infostate_\pl(h) = \infostate_\pl} \reachProb^\policy_\others(h)$.

\smallskip

To define $v$- and $q$-\defword{values} for $h \in \histories$ and $a_\pl \in \actions_\pl(h)$, we set\footnote{
    For a more detailed explanation (incl. the treatment of the case where $P^\pi_{-i}(s_i) = 0$), see \cite{seitz2021learning}.
}
\begin{align*}
     v^\pi_i(h) & :=
        \sum \left\{
            \reachProb^\policy(z) u_\pl(z)
            \mid 
            z \in \mc Z, \ z \suffix z
        \right\}
        /
        \reachProb^\policy(h)
     \\
     \historyActionValue^\policy_\pl(h, \action_\pl) & :=
        \sum \left\{
            \reachProb^\policy(z) u_\pl(z)
            \mid 
            z \in \mc Z, \ z \suffix h a_\pl
        \right\}
        /
        \reachProb^\policy(h)
        .
\end{align*}
We then extend this notation to infostate- and infostate-action values:
\begin{align*}
    \infostateValue^\pi_i(s_i) & := \sum_{s_i(h)=s_i} P^\pi(h \mid s_i) v^\pi_i(h),
    \\
    \infostateActionValue^\pi_i(s_i,a) & := \sum_{s_i(h)=s_i} P^\pi(h \mid s_i) q^\pi_i(h,a_i) ,
\end{align*}
    where $P^\pi(h \mid s_i) = P^\pi(h) / \reachProb^\policy(\infostate_\pl)$.
Finally, we define \defword{counterfactual values} as the (non-counterfactual) values multiplied by counterfactual reach-probabilities:
\begin{align*}
    \historyValue^\pi_{i,\textnormal{cf}}(h)
        & := P^\pi_{-i}(h) \historyValue^\pi_i(s_i) \\
    \historyActionValue^\pi_{i,\textnormal{cf}}(s_i,a_i)
        & := P^\pi_{-i}(h) \historyActionValue^\pi_i(s_i,a_i) \\
    \infostateValue^\pi_{\textnormal{cf}}(s_i)
        & :=
        \sum_{s_i(h) = s_i} \historyValue^\pi_{i,\textnormal{cf}}(h) \\
    \infostateActionValue^\pi_{\textnormal{cf}}(s_i,a_i)
        & :=
        \sum_{s_i(h) = s_i} \historyActionValue^\pi_{i,\textnormal{cf}}(h, a_i)
    .
\end{align*}
This allows us to define \defword{counterfactual regret} (at $\infostate_\pl \in \infostates_\pl$, under $\policy \in \policies$) as
\begin{equation*}
\regret^\policy_\cf{\pl} (\infostate_\pl, \action_\pl)
: =
    \infostateActionValue^\policy_\cf{\pl} (\infostate_\pl, \action_\pl)
    -
    \infostateValue^\policy_\cf{\pl} (\infostate_\pl)
.
\end{equation*}

\smallskip

The final ingredient needed for CFR is a way of translating regrets into strategy-updates.
We do this using the standard \defword{regret matching} (\regretMatching) formula at each infostate
    (where $x^+ := \max \{ x, 0 \}$
    and $\policy^{t+1}_\pl( \action_\pl \mid \infostate_\pl) := 1 / \lvert \actions_\pl(\publicState)\rvert$ when the denominator is 0):
\begin{equation*}
    \policy^{t+1}_\pl ( \action_\pl \mid \infostate_\pl )
    :=
    \regret^{t,+}_\cf{\pl}(\infostate_\pl, \action_\pl)
    \big/
    \sum\nolimits_{\actionAlt_\pl \in \actions_\pl(\infostate_\pl)}
        \regret^{t,+}_\cf{\pl}(\infostate_\pl, \actionAlt_\pl)
    .
\end{equation*}

With these tools, defining \defword{CFR} is straightforward (Algorithm~\ref{alg:vanilla-cfr}):
We initialize the algorithm with a uniformly random policy profile $\policy^0$.
At each iteration, we calculate the counterfactual regrets of $\policy^t$ for all infostates via some method \regUpdate{} and use them to update the policy via regret matching.
Finally, we return the average of the strategies $\policy^t$.

\medskip

In practice, we compute counterfactual values in a single forward- and backward-pass through the game tree.
One way of achieving this is to use the following formulas, corresponding to a ``history-based'' implementation (\Cref{alg:cfv-update-hist}):
\begin{align}
    \historyValue^\policy_{\pl \cf{}}(z)
        & =
        \reachProb^\policy_\others(z)
        \utility_\pl(z)
        \ \ \ \ \textnormal{ for } z \in \mc Z
        \nonumber{}
        \\
    \historyActionValue^\policy_{\pl \cf{}}(h, \action_\pl)
        & =
        \sum_{\action_\others \in \actions_\others(h), \, w \in \mc W}
            \historyValue^\policy_{\pl \cf{}}(haw)
            \label{eq:regret-backprop-histories}
        \\
    \historyValue^\policy_{\pl \cf{}}(h)
        & =
        \sum_{\action_\pl \in \actions_\pl(h)}
            \policy_\pl(\action_\pl \mid h)
            \historyActionValue^\policy_{\pl \cf{}}(h, \action_\pl)
        \nonumber{}
    .
\end{align}
On the way down, we incrementally compute the counterfactual reach probabilities of each history, until we get $\reachProb^\policy_\others(z)$.
    This allows us to compute the counterfactual values of leaves, which we then back-propagate during the upwards pass.
For the purpose of this text, we refer to this history-based implementation as \vanillaCFR{}.
Since \vanillaCFR{} inspects every element of $\histories$, its per-iteration run-time complexity is $\bigO(\lvert\histories\rvert)$:

\begin{restatable}{proposition}{vanillaCFRcomplexity}\label{thm:Vanilla-CFR-complexity}
The time complexity of one \vanillaCFR{} iteration is $\bigO(\lvert\histories\rvert)$.
\end{restatable}

\noindent
The memory complexity is lower-bounded by $\lvert \infostates \rvert$, the number of infostates (because of the need to store the current policy).
However, if \vanillaCFR{} stores the whole $\histories$ in memory, the actual memory complexity can be higher.

\begin{algorithm}
    \begin{algorithmic}[1]
    \State $\policy^0 \gets \text{uniform random policy}$
    \State $\regret(\infostate_\pl, \action_\pl) \gets 0
        \quad \forall \pl \leq N \forall \, \infostate_\pl \in \infostates_\pl \, \forall \action_\pl \in \actions_\pl(\infostate_\pl)$
    \For {$t = 0, \dots, T-1$}
        \State $\regUpdate(\gameroot)$
        \For {$\pl \leq N$ and non-terminal $\infostate_\pl \in \infostates_\pl$}
            \State $\policy^{t+1}_\pl( \, \cdot \, \mid \infostate_\pl) \gets \regretMatching(\regret( \infostate_\pl, \, \cdot \, ))$
        \EndFor
    \EndFor
    \State \Return $\bar \policy = \frac{1}{T} ( \policy^1 + \dots + \policy^T )$
    \end{algorithmic}
    \caption{$\CFR$ using a particular $\regUpdate$} \label{alg:vanilla-cfr}
\end{algorithm}

\begin{algorithm}
    \begin{algorithmic}[1]
    \caption{
        $\histRegUpdate(\history)$.
        \hspace{14em}$\ $ 
        \textit{Implementation of the \vanillaCFR{} regret update on the history tree.}
    }\label{alg:cfv-update-hist}
    \If{$\history = \gameroot$}
        \State $\regret(\infostate_\pl, \action_\pl) \gets 0$
            \hspace{2em} $\forall \pl \leq N \, \forall \infostate_\pl \in \infostates_\pl \, \forall \action_\pl \in \actions_\pl(\infostate_\pl)$
        \State $\reachProbCode_\pl(h) \gets 1$  \hspace{2em} $\forall \pl \leq N$
    \EndIf
    \If{$\history = \leaf \in \leaves$}
        \State \Return $
            \left( \,
                (
                    \prod_{\opp \neq \pl}
                        \reachProbCode_j(\infostate_\opp)
                )
                \reachProbCode_\chance(z)
                u_\pl(z)
            \, \right)_{\pl=1}^N
        $
    \Else
        \State $q_\pl(\action_\pl) \gets 0$
            \hspace{2em} $\forall \pl \leq N \, \forall \action_\pl \in \actions_\pl(h)$
        \For{$\action = (\action_1, \dots, \action_N) \in \actions(\history)$ and $w \in \supp(\mc T(\history, \action))$}
            \State $\reachProbCode_\pl( \history \action w ) \gets \reachProbCode_\pl( \history ) \policy_\pl( \action_\pl \mid \history)
                \hspace{2em} \forall \pl \leq N$
            \State $(q_1(\action), \dots, q_N(\action)) \gets \histRegUpdate(\history \action w)$
            \State $q_\pl(\action_\pl) \increment q_\pl(\action)
                \hspace{2em} \forall \pl \leq N$       
        \EndFor
        \State $v_\pl \gets \sum_{\action_\pl \in \actions_\pl(\history)} \policy_\pl(\action_\pl \mid \history) q_\pl(\action_\pl)$
            \hspace{2em} $\forall \pl \leq N$
        \State $\regret( \infostate_\pl(\history), \action_\pl) \increment q_\pl(\action_\pl) - v_\pl$
            \hspace{2em} $\forall \pl \leq N \ \forall \action_\pl \in \actions_\pl(h)$  
        \State \Return $(v_1, \dots, v_N)$
    \EndIf
    \end{algorithmic}
\end{algorithm}
\section{Public-State CFR}\label{sec:ps-cfr}

\algnewcommand{\LineComment}[1]{\item[] \(\triangleright\) \emph{#1}}
\algrenewcomment[1]{\hfill \(\triangleright\) \ \emph{#1}}

\begin{algorithm}[t]
\caption{$\psRegUpdate(\publicState)$}
\begin{algorithmic}[1]

\LineComment {at the root, initialize reach probabilities by $1$}
\If {$\publicState = \gameRoot$ is the root}
    \State {
        $\reachProb_\pl(\privState_\pl, \gameRoot) \gets 1$
        \quad $\forall \privState_\pl \in \infostates_\pl(\gameRoot)$ $\forall \pl \leq N$
    }
\EndIf

\item[]
\LineComment{once we have reach probs in leaves, use them to compute cf. values} 
\If {$\publicState = \publicTerminalState$ is terminal}
    \For {$\privState_\pl \in \privStates_\pl(\publicTerminalState)$, $\pl \leq N$}
        \label{line:ps:terminal-eval-line1}
        \State $
            \infostateValue_\cfSubscript ( \privState_\pl \mid \publicTerminalState)
            =
            \!\!\!\!\!\!\!\!\!\!
            \sum\limits_{
                    (\privState_\opp)_{\opp \neq \pl}, \,
                    \privState_\opp \in \privStates_\opp(\publicTerminalState)
                }
            \!\!
                \left(
                    \prod\limits_{\opp \neq \pl}
                        \reachProb_\opp( \privState_\opp, \publicTerminalState)
                \right)
                \chanceWeightedUtilities_\pl(\privState_\pl, \privState_\opp \mid \publicTerminalState)
        $
            \label{line:ps:terminal-eval}
    \EndFor
\EndIf

\item[]
\LineComment{propagate reach probabilities downwards and cf. values upwards} 
\If {$\publicState$ is not terminal}
    \For {
        $\privState_\pl \in \privStates_\pl(\publicState)$,
        $\pl \leq N$
    }
        \Comment {$\downarrow$-pass: set action-vals to $0$}
        \State {
            $\infosetActionValue_\cfSubscript(\privState_\pl, \action_\pl \mid \publicState) \gets 0$
                \quad $\forall \action_\pl \in \actions_\pl(\privState_\pl \mid \publicState)$
        }
    \EndFor
    \For {$\publicStateAltAlt \in \ims(\publicState)$}
        \Comment{go through immediate successors}
        \For {$\privStateAlt_\pl = \privState_\pl \action_\pl \observation_\pl \in \privStates_\pl(\publicStateAltAlt)$, $\pl \leq N$}
            \Comment {$\downarrow$-pass: compute reach probs}
            \State {$
                    \reachProb_\pl (\privStateAlt_\pl, \publicStateAltAlt)
                    \gets
                    \reachProb_\pl (\privState_\pl, \publicState)
                        \policy_\pl (\action_\pl \vert \privState_\pl, \publicState)
                $
            }
        \EndFor
        \State {$\psRegUpdate(\publicStateAltAlt)$}
        \For {for all $\privStateAlt_\pl = \privState_\pl \action_\pl \observation_\pl \in \privStates_\pl(\publicStateAltAlt)$, $\pl \leq N$}
            \State {$
                    \infostateActionValue_\cfSubscript(\privState_\pl, \action_\pl \mid \publicState)
                    \increment
                    \infostateValue_\cfSubscript(\privStateAlt_\pl \mid \publicStateAltAlt)
                $
            }
                \Comment {$\uparrow$-pass: update action-vals}
                \label{line:ps:Q-update}
        \EndFor
    \EndFor
    \For {$\privState_\pl \in \privStates_\pl(\publicState)$, $\pl \leq N$}
        \Comment {$\uparrow$-pass: update cf. regrets}
        \State $\infosetValue_\cfSubscript(\privState_\pl \mid \publicState)
            \gets
            \sum_{\action_\pl \in \actions_\pl (\privState_\pl \mid \publicState)}
                \policy_\pl(\action_\pl \mid \privState_\pl, \publicState)
                \infosetActionValue_\cfSubscript(\privState_\pl, \action_\pl \mid \publicState)
        $
            \label{line:ps:V-update}
        \State $
            \regret(\privState_\pl, \action_\pl \mid \publicState)
            \increment
            \infosetActionValue_\cfSubscript(\privState_\pl, \action_\pl \mid \publicState)
                - \infosetValue_\cfSubscript(\privState_\pl \mid \publicState)
        $
            \label{line:ps:regret-update}
    \EndFor
\EndIf
%
\end{algorithmic}
\end{algorithm}

In \Cref{sec:background}, we described an implementation of the CFR algorithm which runs on the history tree of the game.
In this section, we describe an implementation of CFR that traverses the game on the level of information states.
    We also show that this implementation is amenable to parallelisation by decomposing the computation based on public information
            (that is, if multiple infostates correspond to the same sequence public observations, CFR visits them at the same time).
    For this reason, we call this implementation \defword{public state CFR} (\publicStateCFR{}).

Analogously to the formula \eqref{eq:regret-backprop-histories},
    counterfactual values can be computed by backpropagation over information states \cite[Thm.~2,\,(4)]{seitz2023value}:

\begin{restatable}[CFV computation over infostate tree]{lemma}{cfvInfostateBackprop}\label{lem:CFV-update-infostates}
For any $\policy \in \policies$, we have:
\begin{align}
    \infostateValue^\policy_{\pl \cf{}}(\infostate_\pl)
        & =
        \sum_{z \in \mc Z, \, \infostate_\pl(z) = \infostate_\pl}
            \reachProb^\policy_\others(z)
            \utility_\pl(z)
            \ \ \ \ \textnormal{ for terminal infostates }
            \label{eq:regret-backprop-infostates-terminal}
        \\
    \infostateActionValue^\policy_{\pl \cf{}}(\infostate_\pl, \action_\pl)
        & =
        \sum_{\observation_\pl \in \mb O_\pl}
            \infostateValue_{\pl \cf{}}(\infostate_\pl \action_\pl \observation_\pl)
            \label{eq:regret-backprop-infostates-Q}
        \\
    \infostateValue^\policy_{\pl \cf{}}(h)
        & =
        \sum_{\action_\pl \in \actions_\pl(h)}
            \policy_\pl(\action_\pl \mid \infostate_\pl)
            \infostateActionValue^\policy_{\pl \cf{}}(\infostate_\pl, \action_\pl)
        \label{eq:regret-backprop-infostates-V}
    .
\end{align}
\end{restatable}


\noindent
To run this computation fully over infostates, we additionally make the following observation:
Recalling that every infostate $\infostate_\pl$ can be decomposed into a private and public part as $\infostate_\pl = (\privState_\pl, \publicTerminalState)$,
we denote by
\begin{align*}
    \infostates_\pl(\publicState) & := \{ \infostate_\pl \in \infostates_\pl \mid \infostate_\pl = (\privState_\pl, \publicState) \} \\
    \privStates_\pl(\publicState) & := \{ \privState_\pl \mid (\privState_\pl, \publicState) \in \infostates_\pl \}
\end{align*}
    the set of information states, resp. private information states, compatible with a given public state $\publicState$.
This allows us to rephrase the formula \eqref{eq:regret-backprop-infostates-terminal} as follows.

\begin{restatable}{lemma}{cfvViaCWU}\label{lem:chance-weighted-utilities}
For every
    $\policy \in \policies$
    terminal $\publicTerminalState \in \publicTree$, and
    $\infostate_\pl \in \infostates_\pl(\publicTerminalState)$,
    we have
\begin{align*}
 \infostateValue^\policy_{\pl \cf{}}(\infostate_\pl)
        & =
        \sum_{
            \overset{
                (\privState_\opp)_{\opp \neq \pl}
                \ \in 
            }{
                \prod_{\opp \neq \pl}
                    \privStates_\pl(\publicTerminalState)
            }                
        }
            \left(
                \prod_{\opp \neq \pl, \chance}
                    \reachProb^\policy_\opp(\privState_\opp, \publicTerminalState)
            \right)
            \sum_{
                \overset{
                    z \in \mc Z
                \, : \,
                \privState_\opp(z) = \privState_\opp
                }{
                    \, \forall \opp = 1, \dots, N
                }
            }
                \!\!\!\!\!
                \reachProb_\chance(z)
                \utility_\pl(z)
            \\
        & =:
        \sum_{
            \privState_\others
            \in
            \privStates_\others(\publicTerminalState)
        }
            \left(
                \prod_{\opp \neq \pl, \chance}
                    \reachProb^\policy_\opp(\privState_\opp, \publicTerminalState)
            \right)
            \chanceWeightedUtilities_\pl(\privState_\pl, \privState_\others \mid \publicTerminalState)
    .
\end{align*}
\end{restatable}

\noindent
For example, in a two-player game, the interpretation of \Cref{lem:chance-weighted-utilities} is that
    every terminal public state is associated with a matrices
        $\chanceWeightedUtilities_\pl( \, \cdot \, , \, \cdot \, \mid \publicTerminalState)$,
        $\pl = 1, 2$,
    that are indexed by $\privState_1 \in \privStates_1(\publicTerminalState)$, $\privState_2 \in \privStates_2(\publicTerminalState)$.
Correspondingly, each cell of the matrix
    groups together all histories that would be indistinguishable even if the players pooled together their private information and
    captures them in the single number, called \defword{chance-weighted utility:}
    \begin{align*}
        \chanceWeightedUtilities_\pl(\privState_1, \privState_2 \mid \publicTerminalState)
        :=
        \sum \left\{ 
            \reachProb_\chance(z)
            \utility_\pl(z)
            \mid
            z \in \mc Z, \, \privState_1(z) = \privState_1, \, \privState_2(z) = \privState_2
        \right\}
        .
    \end{align*}
To obtain the vector of counterfactual values $(\infostateValue^\policy_{\cf{1}}(\privState_1))_{\privState_1 \in \privStates_1(\publicTerminalState)}$,
    all we need to do is to multiply this matrix by the vector of reach probabilities $(\reachProb^\policy_2(\privState_2))_{\privState_2 \in \privStates_2(\publicTerminalState)}$.

By putting these results together (\Cref{alg:cfv-update-hist}), we get a method \psRegUpdate{} which updates counterfactual regrets by traversing the game tree on the level of public states.
We refer to the version of CFR which updates regrets via \psRegUpdate{} as \defword{public-state CFR} (\publicStateCFR{}).
While the following should be clear from the formulation of the algorithm, it is worth stating explicitly.
    \publicStateCFR{} and \vanillaCFR{} are merely two \textit{different implementations of the same algorithm}
    --- that is, their per-iteration time- and memory- complexity might be different, but \textit{their outputs are identical}:

\begin{restatable}{proposition}{PSCFRisEquivalent}\label{prop:vanilla-and-ps-are-equivalent}
For any $\game$ and $T \in \N$,
    the strategy produced by $T$ iterations of $\publicStateCFR{}$ is the same as the strategy produced by $T$ iterations of $\vanillaCFR{}$.
\end{restatable}

The main advantage of \publicStateCFR{} over \vanillaCFR{} is its potential for increased practical and asymptotic efficiency,
    which we discuss in the remainder of this section.
\subsection{Asymptotic Complexity of Public State CFR}\label{sec:sub:complexity-theory}

The theoretical results around \publicStateCFR{}'s complexity are summarized in \Cref{thm:PSCFR-complexity} below.
The intuitions for these results are as follows:
The lower bound follows from the need to store and update the policy (at least for tabular implementations).
    This means that the time and space complexity of one \publicStateCFR{} iteration (i.e., of \psRegUpdate{}) cannot be lower than,
        $
            \lvert
                \bigcup_{\pl=1}^N
                    \infostates_\pl
            \rvert
            =:
            \lvert \infostates \rvert
        $
        the number of information states in the game.
The upper bound holds because
    $\game$ cannot have more information states (or legal ``infostate-profiles'' $(\privState_1, \dots, \privState_N)$) than histories.
    However, these amounts can be similar or even equal (as in perfect-information games).
Finally, as we discuss below,
    the generic evaluation of terminal states can often \cite[Ex.\,1-3]{acceleratedBR} be replaced by a more efficient domain-specific implementation.

\begin{restatable}[Complexity of \publicStateCFR{}]{theorem}{PSCFRcomplexity}\label{thm:PSCFR-complexity}
The time and space complexity of one iteration of \publicStateCFR{} satisfies:
\begin{enumerate}[label=(\arabic*)]
    \item The complexity is always at least $\lvert\infostates \rvert$ and at most $\bigO(\lvert \histories \rvert)$.
    \item Treating $N$ and $\chanceWeightedUtilities_\pl$ as constants, the complexity is $\bigO( \lvert \publicTree \rvert \prod_{\pl=1}^N \lvert \privStates_\pl \rvert )$.\footnotemark
    \item In some domains, including poker, the $\publicTerminalState$ evaluation (lines \ref{line:ps:terminal-eval-line1}-\ref{line:ps:terminal-eval}) can be implemented s.t. the overall complexity is
        $
            \bigO( \lvert \publicTree \rvert \, \lvert \bigcup_{\pl=1}^N \privStates_\pl \rvert )
            =
            \bigO( \lvert \infostates \rvert )
        $.
\end{enumerate}
\end{restatable}
    \footnotetext{
        Recall that $\privStates_\pl$ denotes the set of all possible private observations of Player $\pl$.
    }

Poker serves as a good illustration of (3):
The key insight is that if we want to know the expected utility in poker, it is unnecessary to know the probability of each card combination (hand) that the opponent could have.
Instead, we only need to know the probability that their hand is weaker than ours.
Moreover, suppose our hand changed to a stronger one.
In that case, the probability of the opponent's hand being weaker would only change by the probability of the opponent holding cards weaker than our new hand but stronger than our original hand.
This observation allows us to re-use most of the computation, bringing the per-public-state complexity from $\lvert\privStates_1\rvert  \lvert \privStates_2 \rvert$ to $\lvert \privStates_1 \rvert +\lvert \privStates_2 \rvert$ \cite{acceleratedBR}.
The same procedure can be applied to any game with a notion of ``private-state strength'' (i.e., a linear ordering of $\privStates_\pl$) for any terminal public state.

A reduction in complexity will also be possible in many other games.
For example, in both battleship and stratego, once we reach a terminal public state, the utilities become completely determined by public information!
(In battleship, the player with all ships destroyed loses the game, so the position of the other player's ships can be ignored.
In stratego, the player whose flag is captured loses the game, so the position of all other pieces becomes irrelevant.)
Informally speaking, the potential for complexity reduction is larger when many different private-state combinations $(\privState_1, \dots, \privState_N)$ lead to identical payoffs.
    (As an example of a non-reducible game, we can consider the contrived version of battleship where each joint position of all ships is associated with a unique bonus payoff awarded to the winner.)

\subsection{Practical Improvements of Public State CFR}\label{sec:sub:complexity-practice}

We now discuss two practical ways of improving \publicStateCFR{}.

First, we should note that some games have a more regular structure than others, which in turn allows for \textit{simpler implementations} of \publicStateCFR{}.
To see this, note that the algorithm \publicStateCFR{} implicitly requires the following structures:
    To traverse the tree, we need access to a list of immediate successors of every public state and we need to know the predecessor of every information state (for \cref{line:ps:Q-update}).
    In various places, we need to know the list of private infostates $\privStates_\pl(\publicState)$ compatible a given public state $\publicState$.
    And finally, for any $\privState_\pl \in \privStates_\pl(\publicState)$, we need to know which actions are legal given $\privState_\pl$ at $\publicState$.
In general, each of these structures might require complicated custom implementations or be computationally expensive to build.
    For example, in blind chess or dark chess, determining the actions available to the opponent is very non-trivial and so is figuring out which information states are even possible for them.
Conversely, if the structure of the game is more regular, the implementation of \psRegUpdate{} can be correspondingly simpler.
    A canonical example of such simpler game is limit hold'em poker:
        In this game, the set of available actions is essentially the same at every state.
        Moreover, the only private information is generated at the start of the game (when the players draw private cards),
        so at any point of the game, we can assume that $\privStates_\pl(\publicState) \subset \privStates_\pl = \{ \textnormal{possible private-card combinations} \}$.

Second, if the game structure is sufficiently regular, most of the operations in \psRegUpdate{} can be implemented using vector and matrix operations.
In particular, as already mentioned above, the terminal-state evaluation (lines \ref{line:ps:terminal-eval-line1}-\ref{line:ps:terminal-eval}) can be implemented as a multiplication of the matrix $\chanceWeightedUtilities_\pl(\, \cdot \, , \, \cdot \, \mid \publicTerminalState)$ by the vector of reach probabilities.
In practice, this makes the algorithm more amenable to efficient implementations than \vanillaCFR{}.

A particular class of games which can be especially amenable to these practical improvements is what we could call \BGFull{}s:
\section{Sequential Bayesian Games}\label{sec:beg}

In this section, we describe the class of sequential Bayesian games (SBGs).
    Our primary interest in SBGs is that they have a very regular structure, which makes them more amenable to \publicStateCFR{} than general games.
Our second motivation revolves around the observation that the examples considered in the CFR literature very often belong precisely to this class of games, without this fact being widely recognised.
    We thus find it valuable to bring this class of games to the attention of the CFR community more explicitly.

To build the intuition for SBGs, recall that a (non-sequential) Bayesian game \cite{zamir2020bayesian} is a game where
    each player has a private type (drawn from a known joint distribution),
    the players select their actions simultaneously (for one round only),
    and the resulting utilities depend on the combination of the joint action and the joint type.
A sequential Bayesian game should thus work analogously, except for allowing sequential interactions.

Formally, we consider the following definition:

\begin{definition}[Sequential Bayesian game]\label{def:sequential-BG}
A factored-observation stochastic game $\game$ is a \defword{sequential Bayesian game} (SBG) if
\begin{enumerate}[label=(\roman*)]
    \item apart from the initial randomisation, there are no private observations:\\
        $
            \forall (w, a, w') \, \forall \pl, \opp \in \mc N:
            w \neq \initState
            \implies
            \mc O_\pl(w, a, w') = \mc O_\opp(w, a, w')
        $;
    \item whether a player's action is legal only depends on public information\footnotemark{}:\\
        $
            \forall h, h' \in \histories \, \forall \pl \in \mc N :
            \publicState(h) = \publicState(h')
            \implies
            \actions_\pl(h) = \actions_\pl(h')
        $;
        \label{item:legal-actions}
    \item and all player-actions are publicly observable:\\
        $
            \forall (w, a, w'), (w, \hat a, \hat w') \, \forall \pl \in \mc N :
            a_\pl \neq \hat a_\pl
            \implies
            \mc O_\public(w, a, w') \neq \mc O_\public(w, \hat a, \hat w')
        $.
\end{enumerate}
\end{definition}

\footnotetext{
    Strictly speaking, the condition \eqref{item:legal-actions} is redundant
        since it allows us to make actions effectively illegal for specific private-states by having them incur a prohibitively high reward penalty.
}

\noindent
An important feature of this definition is that it
    (a) allows the initial private observations (i.e., private types) to be correlated,
    (b) allows the existence of ``hidden features'' of the game state (as long as no player has information about them), and
    (c) allows the stochastic transitions in $\game$ to depend on private types and hidden features.
This allows us to model games such poker,
    where the players' private cards, as well as the later-revealed public cards, come from the same deck (and are thus correlated).
Moreover, (a-c) are closely connected, such that getting rid of any of them might require getting rid of all of them.
    (For example (a) allows us to implement (b) and (c) by including dummy players who use a fixed strategy.)
In particular, one could consider a version of \Cref{def:sequential-BG} which additionally
    requires the initial observations to be independent
    and the stochastic transitions to only depend on public information.
However, such additional assumptions would likely fatally limit the formalism's ability to capture practically interesting settings.

Similarly to normal- and extensive-form games,
    \defword{being a sequential Bayesian game is a property of a particular \textit{representation} of a real-world scenario,
    not a property of the scenario itself}:
For example,
    poker is (arguably) naturally modelled as a FOSG  that satisfies \Cref{def:sequential-BG}
    while matching pennies is (arguably) naturally\footnotemark{} modelled as a FOSG that does not satisfy it.
However,
    the FOSG representation of poker can be trivially modified to fail \Cref{def:sequential-BG}
    while the normal-form representation of matching pennies (viewed as a FOSG) satisfies the definition.
In fact, \textit{any} game can be cast as a SBG:
    This can always be done trivially, by using the game's normal-form representation.
    However, as suggested by \Cref{ex:SB-form-representation} below, it might even be possible to find a sequential Bayesian representation that preserves the natural sequential structure of the game
        (though this will come at the cost of having to consider a potentially huge number of auxiliary types for each player).
In summary,
    these observations suggests that it does not make sense, for example, to formally ask whether ``poker is a sequential Bayesian game''.
    Instead, we can either ask
        the formal question ``is a given model of poker is a SBG''
        or the informal question ``is poker is naturally modelled as a SBG''.

\footnotetext{
    Recall that the motivating story behind matching pennies is that one player starts by placing a coin with heads/tails facing up,
    and the other player only makes their guess afterwards.
    In other words, the first player's action at time $t=0$ is not observable by the second player at $t=1$, which is very much \textit{not} allowed in a sequential Bayesian game.
}

There are several well-known examples of games which are natural to model as SBGs.
    The two that appears the most in the CFR literature are
        hold'em poker \cite{DeepStack}
        and liar's dice \cite{liarsDiceWiki} (and its variants Bluff, Dudo, and others).
    Another example is the game battleship \cite{battleshipWiki} if we assume that the placement of ships is determined using a fixed probability distribution.
    The same would apply to stratego \cite{strategoWiki} (if all pieces moved the same).

Finally, the following example suggests a method for converting ``non-\BGAbbrev{}s'' to their ``SB-form representation''.
While we don't expect such representations to be of practical importance, we find them relevant for understanding the formalism's expressive power.

\begin{example}[Sequential-Bayesian-form representation of general games]\label{ex:SB-form-representation}
To illustrate the idea on a simple game, consider the variant of matching pennies where:
    At time $t=0$, player one (P1) selects heads ($H$) or tails ($T$), without player two (P2) being aware of their choice
    At $t=1$, P2 guesses either $H$ or $T$.
    At $t=2$, P1's initial choice is revealed and either P1 pays $\$1$ to P2 if P2 guessed correctly, or P2 pays $\$1$ if P2 guessed incorrectly.
To turn the game into a SBG, we modify it as follows:
    At $t=-1$, P1 privately receives a ``code book'' ---
        formally, a bijection between $ \varphi: \{ H, T \} \to \{ X, Y \}$,
        sampled from the uniformly random distribution over all such bijections.
    At $t=0$, P1 internally selects $H$ or $T$, encodes it using $\varphi$, and publicly announces the result.
        Formally, P1 just selects $\action_1 \in \{ X, Y \}$ and P2 observes this choice.
    At $t=1$, the game remains unchanged --- P2 selects $\action_2 \in \{ H, T \}$ and P1 observes this choice.
    Finally, at $t=2$, the game engine decodes P1's action and awards utilities as in the original game.
        Formally, the players receive utility $\pm 1$ depending on whether P2's guess $\action_2$ matches $\varphi^{-1}(\action_1)$.

More generally, we can assume that at the start of the game, each player receives a sequence of randomly selected ``code books'', one for each round of the game.
Afterwards, every time one of the players needs to take an action that would not be public in the original game,
    they encode it using their code book for the current round and announce the result publicly.
Similarly, an analogous procedure applies to observations that are private in the original game.
    (That is, each player receives a randomly generated observation code-book for each round.
    All originally-private observations are first encoded using a corresponding code book and then announced publicly.
    This ensures that the modified game satisfies the definition of SBG while giving each player precisely as much information as they had in the original game.)
\end{example}
\section{Empirical Evaluation}\label{sec:empirical}

In the previous section, we have shown that in poker, \publicStateCFR{} with a domain-specific evaluation of terminal states has asymptotically lower run-time and memory complexity than \vanillaCFR{}.
We also argued that \publicStateCFR{} is likely to be more efficient than \vanillaCFR{} even without the domain-specific optimizations.
We now compare the performance of these algorithms empirically.

We evaluate the results on a subgame of no-limit Texas hold'em poker. Specifically, we use the subgame after the last public card is dealt (\ie a river subgame), with public cards (9s, 7c, 5s, 4h, 3c), pot size 200, and uniform distribution over private cards. To make the computation tractable, we use the (fold, call, pot, all-in) action-abstraction, which results in a subgame that has 61,000,831 states (i.e., histories) and 21,620 decision points (i.e., active-player infosets).

All algorithms are implemented using the open-source library OpenSpiel \cite{openspiel}.
All computations ran on a single thread of a single CPU (Intel Xeon Gold 5120 2.2GHz).
    While public-state CFR is amenable to parallelization, we stick to a single-thread implementation running on a CPU to make the comparison to \vanillaCFR{} meaningful.
We implemented two versions of public-state CFR --- the baseline version \genericPSCFR{} applicable to any \BGAbbrev{} and a poker-specific version \pokerPSCFR{} whose terminal-state evaluation run-time is linear (rather than quadratic) in the number of infosets \cite{acceleratedBR}.
For vanilla CFR, we used a version already present in OpenSpiel.
We refer to this version as \openSpielCFR{} since it only maintains a small portion of the game tree at any given time, resulting in slower child-retrieval but low memory usage.
(Recall that the memory requirements of CFR are lower-bounded by the need to store the strategy in each infoset.
An analogous low-memory variant is thus unnecessary for \genericPSCFR{}, whose memory requirements are already quite low; see Table~\ref{tab:evaluation}).
To make the algorithm more directly comparable to our implementation of public-state CFR, we also implemented a version that keeps the structure needed for child-retrieval in memory.
We refer to this version simply as \saveCFR{}.

\begin{table*}[tb]
\begin{center}
\begin{tabular}{|c || c | c | c | c |}
 \hline
 Algorithm & Setup & 1000 it. & One it. & Memory\\
 \hline\hline
 \saveCFR{} \phantom{(mem-eff)} & 2.92 min & 5.48 h\phantom{in}\hspace{0.3em} & 19.73 s\hspace{0.4em} & \hspace{0.5em}22 GB\\
 \hline
 \phantom{Vanil} \genericPSCFR{} \phantom{(mem-eff)} & 2.36 s\phantom{in}\hspace{0.5em} & 2.89 min &  173.15 ms & 736 MB\\
 \hline
 \phantom{Vanil} \pokerPSCFR{}\phantom{f} & 2.82 s\phantom{in}\hspace{0.5em} & 1.42 min & \hspace{0.5em}85.05 ms & 526 MB\\
 \hline
 \phantom{}\openSpielCFR{}\phantom{} & 1.57 min & 25.75 h \phantom{in}\hspace{0.4em} & 92.71 s\hspace{0.35em} & 292 MB\\
 \hline
\end{tabular}
\end{center}
\caption{A comparison of \vanillaCFR{} and \publicStateCFR{} on a river subgame of no-limit Texas hold'em poker.
}
\label{tab:evaluation}
\end{table*}
\begin{figure}
    \centering
    \includegraphics[width=0.9\textwidth]{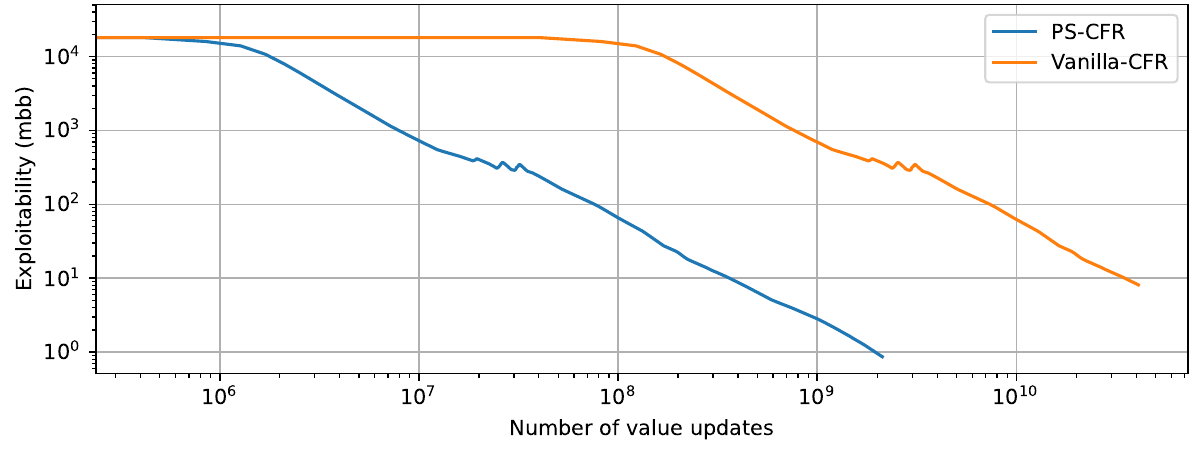}
    \caption{Comparison of necessary value updates in Vanilla CFR and PS-CFR. Note that singe iteration takes roughly $4 \times 10^5$ for PS-CFR and $4 \times 10^7$ for Vanilla CFR.}
    \label{fig:num_updates}
\end{figure}

We ran 1000 iterations of each algorithm and measured the resulting memory usage, the time required for initialization, and the subsequent time needed to run the 1000 iterations (Table~\ref{tab:evaluation}).
First, we see that the poker-specific version of \publicStateCFR{} takes slightly longer to initialize, but afterwards only requires 70\% of the memory and runs twice as fast.
Among the two versions of vanilla CFR, \openSpielCFR{} is roughly five times slower than \saveCFR{} but uses roughly 75x less memory.
Most importantly, we see that both versions of vanilla CFR are extremely slow compared to \publicStateCFR{}:
The faster version requires roughly 200x more time and 40x more memory than \pokerPSCFR{},
while the slower version requires roughly 1000x more time and a similar amount of memory.

One worry is that the experiment above could be misleading by, for example,
    comparing a highly optimized version of PS-CFR with a simplistic implementation of Vanilla-CFR.
To compare the two algorithm in a way that is less sensitive to code optimization,
    we also look at the number of value updates
    --- that is, the number of times that each algorithm changes the (counterfactual) value of some information state.
In Figure~\ref{fig:num_updates}, we compare the number of value updates needed by the two algorithms to attain specific exploitability levels.
The results validate several expected observations.
Firstly, based on the exploitability metric, it becomes evident that Vanilla-CFR and PS-CFR execute the same algorithm, yielding identical exploitabilities while differing only in terms of efficiency.
Secondly, the figure shows that Vanilla-CFR needs approximately 100x more value updates than PS-CFR to achieve the same level of exploitability.
In particular, this means that the actual 200x speedup (\Cref{tab:evaluation}) is within a factor of $2$ of the speedup caused by the difference in the number of value updates.
This shows that
    the vast majority of PS-CFR's observed advantage over Vanilla-CFR is due to fundamental reasons,
    rather than due to differences in the quality of the particular implementation.
\section{Related Work}
        \label{sec:related-work}

\subsection{Public States and CFR}

The CFR algorithm was first described in \cite{CFR}.
    For a detailed description of the underlying concepts, we refer to the reader to \cite{seitz2023value}.
The full public state formulation of CFR studied in this paper has not been formally described before.
However, many of the related ideas are well-known in the context of poker, and in particular, in the community around the annual poker competition \cite{ACPC}.
In terms of published theory, using public information for more efficient evaluation dates at least to \cite{acceleratedBR}.
    In terms of implementation, public states allow for ``vectorized'' formulation of CFR, which has been used in many papers, including \cite{CFR-D,schmid2019variance,DeepStack,rebel,POG} and parts of \cite{Libratus,Pluribus}.
    (The code of \cite{rebel} is public.
        The implementations of \cite{DeepStack,schmid2019variance} are not public, but they are discussed in the appendix of \cite{schmid2019variance}.
        Authors of \cite{POG} explain in the main paper that they run a variant of a CFR on a public game tree.
        The implementation of \cite{CFR-D} uses the Public Chance Sampling variant of CFR \cite{johanson2012efficient}, which works on infosets.
        The implementation details of \cite{Libratus,Pluribus} are neither public nor discussed in published literature, but our claims about them are at least based on personal correspondence with the authors.)
We thus argue that many papers already use an implementation of CFR that is closer to \genericPSCFR{} (Section~\ref{sec:ps-cfr}) than to \vanillaCFR{}.

Importantly, the concepts around public states have primarily appeared in the context of poker,
    with rare extensions to specific other domains,
    and extending them to general EFGs proved more difficult than one might intuitively expect \cite{kovarik2019problems,FOG}.
Our main contribution in this respect is identifying sequential factored-observation stochastic games (and \BGFull{}s in particular) as the class of games where extending poker results is natural and straightforward.

\medskip

In terms of important but less-directly relevant results,
    note that the literature describes a number of CFR variants such as
    CFR+ \cite{CFR+,revisitingCFR+}, Linear CFR \cite{linearCFR}, or Deep CFR \cite{deepCFR}.
    This line of work aims to improve CFR by \textit{modifying} its output.
As a result, it is orthogonal to -- but likely compatible with -- the public-state reformulation of CFR, which aims to produce the \textit{same} output as vanilla CFR, but to do so more efficiently.

Another line of CFR research is about applying the algorithm to the sequence-form of strategies~\cite{farina2019online}.
    This is a powerful technique
        which reformulates each CFR update as a single matrix-by-vector multiplication,
        and opens the problem to the application of generic matrix multiplication algorithms.
A disadvantage of this approach is that it seems difficult to combine with depth-limited methods~\cite{seitz2023value},
    which proved essential for some of the past successes~\cite{DeepStack}.
At the present time, we are uncertain how the two approaches compare to each other and whether they are compatible or not.

\subsection{Modelling Sequential Bayesian Games}

Sequential variants of Bayesian games are considered, for example, in 
    \cite{osborne1994course} (under the name ``Bayesian extensive game with observable actions''),
    \cite{fudenberg1991game} (as ``multi-stage games with observed actions and incomplete information''), and
    \cite{battigalli2003rationalization} (as ``games of incomplete information with observable actions'').
The models of sequential Bayesian games can vary in several formal details.
First, can different players' information about the rules of the game (i.e., each player's type) be correlated?
    After reading \cite{osborne1994course,fudenberg1991game,battigalli2003rationalization}, we believe that:
        (a) Independent player types simplify some of the analysis, so they are often assumed for convenience.
        (b) However, correlated player types are natural and in line with researchers' intuitions about Bayesian games.
    We allow for correlated player types since it enables us to, for example, model poker, where the players draw their private cards (types) from a shared deck.
Second, is there an explicitly designated ``chance'' player whose policy is fixed?
    The above works don't explicitly include this possibility as a part of the formal definitions but mention it in the discussion.
Finally, can the game dynamics (formally: state transitions) depend on the players' types?
    (Without this dependence, we cannot describe settings such as private and public cards in poker being drawn from the same deck.)
    We allow this since once we have correlated player types, formally adding type-dependent state transitions only affects the notation, not the actual expressive power of the model.
        This is because such transitions could also be described using a chance player who controls the environment and knows each player's type.

\subsection{Imperfect vs Incomplete Information}

The concept of sequential Bayesian games is tied to the informal concept of \emph{incomplete information}.
While some authors use this term interchangeably with \emph{imperfect} information, some also make the following distinction:
A player is said to have \defword{incomplete information} if they are uncertain about the \textit{rules} of the game (\eg legal actions, identity of other players, utility functions).
In contrast, they are said to have imperfect information if they are uncertain about the current \textit{state} of the game.
In his seminal paper \cite{harsanyi1967games}, Harsanyi explains the relationship between the two types of games and introduces (what is now typically called) Bayesian games as a formalization of strategic interaction under incomplete information.
\section{Conclusion}\label{sec:discussion}
    \label{sec:sub:conclusion}

We have recently seen a lot of progress around counterfactual regret minimization.
However, while many works aim to solve general (two-player zero-sum) imperfect information games --- typically formally modeled as extensive-form games --- their empirical evaluation tends only to consider poker, liar's dice, or other games whose structure is near-identical to poker.
As a result, existing implementations of most CFR-based algorithms use many optimizations which rely on non-generalizable poker-specific assumptions.
These optimizations are often so extensive that it might be more appropriate to view the implemented algorithm as distinct from the CFR formulation described in \cite{CFR}.
Moreover,
    since these optimizations were never formally described in a more general setting,
    the ideas are needlessly difficult to translate into other settings,
    particularly for those not already familiar with the CFR literature.

In this paper, we described public-state CFR,
    an algorithm which produces the same outputs as classical CFR (\Cref{prop:vanilla-and-ps-are-equivalent})
    while automatically incorporating some of the optimizations when applied to games with enough public information (cf. \Cref{thm:Vanilla-CFR-complexity} vs \Cref{thm:PSCFR-complexity}).
We argued that
    the algorithm is also amenable to parallelisation
    and can further be made more efficient when the specific domain has compact rules for evaluating the outcome of the game.
Even without these optimizations,
    we saw that with abundant public information,
    PS-CFR significantly outperforms vanilla CFR
        (1000 iterations on a poker subgame taking $\sim\!3$ minutes instead of $\sim\!5.5$ hours; \Cref{sec:empirical}).

We defined a class of \BGFull{}s (\Cref{def:sequential-BG}),
    where each player is initially assigned a private type
        (e.g., private cards in poker),
    and afterwards all actions and observations are public.
While this condition is neither sufficient nor necessary,
    there seems to be a connection between being naturally modelled as a \BGAbbrev{} and being amenable to public-state CFR.
Moreover, it seems noteworthy that hold'em poker and liar's dice
    -- the two historically most popular benchmarks in the CFR literature --
    are both a perfect fit for the \BGAbbrev{} model.

Finally, we believe that these observations shed a new light on the existing CFR results.
First,
    when one implements vanilla CFR to solve some game,
        the algorithm is likely to run significantly slower than one might naively expect from reading the existing literature.
    This is because many of the papers implicitly use something closer to the (already-optimized) public-state CFR.
Second,
    when public observations are abundant,
        it seems beneficial to represent the domain as a factored-observation stochastic game
            (or even \BGFull{} or a similar model).
    One can then apply a generic version of PS-CFR,
        which will be both more efficient than vanilla CFR and
        more reusable than a fully custom implementation.
    (For further speedups, one might, of course, need to apply additional domain-specific improvements.
    Nonetheless, PS-CFR seems as a better starting point than vanilla CFR.)
Finally,
    when public observations are scarce (e.g., in imperfect-information variants of chess),
        both vanilla CFR and PS-CFR are likely to prove overly slow.
    This does not mean that CFR cannot be the right tool for the problem.
    However, it does suggest that to get a viable implementation,
        one likely needs to exploit some additional properties of the given game.
    As a result, dealing with such domains is likely to be more time consuming
        (in terms of researcher time)
        than one might initially hope.

    \subsection*{Acknowledgments}

The initial idea behind \Cref{ex:SB-form-representation} is due to Vincent Conitzer.
This work was supported by the Czech science foundation grant no. GA22-26655S, the Grant Agency of the Czech Technical University in Prague, grant No. SGS22/168/OHK3/3T/13. Computational resources were supplied by the project e-Infrastruktura CZ (e-INFRA CZ LM2018140) supported by the Ministry of Education, Youth and Sports of the Czech Republic.

\bibliography{main}

\begin{appendices}
    \section{Proofs}\label{sec:app:proofs}

\PSCFRisEquivalent*
\begin{proof}
By \Cref{lem:chance-weighted-utilities}, the numbers calculated on \Cref{line:ps:terminal-eval} produce the correct values $\infostateValue^\policy_{\cf{\pl}}$ that satisfy \Cref{eq:regret-backprop-infostates-terminal}.
Similarly, \Cref{line:ps:Q-update} and \Cref{line:ps:V-update} ensure that \Cref{eq:regret-backprop-infostates-Q}, resp. \Cref{eq:regret-backprop-infostates-V} hold.
By \Cref{lem:CFV-update-infostates}, this calculation produces the counterfactual values corresponding to $\policy$, which then yield the correct regrets (\Cref{line:ps:regret-update}).
\end{proof}

\PSCFRcomplexity*

\begin{proof}
(1) The lower-bound holds because each iteration of CFR needs to, at least, update each player's strategy at each information state
    (and keep the whole strategy in memory).
The upper bound follows from the observation that
    at non-terminal public states,
        the algorithm touches either information states $\infostate_\pl \in \infostates_\pl(\publicState)$ or infostate-action pairs (at most $k$-times for some small $k$),
            and the number of each player's infostates $\infostates_\pl$, as well as their legal infostate-action pairs, is smaller than the number of histories $\mc H$.
            (For the infostate-action pairs, this is because the number of legal infostate-action pairs at $\publicState$ is smaller than the number of histories in the \textit{immediate successors} of $\publicState$.)
        The number of operations performed for non-terminal public states is thus smaller than $2k N \lvert \mc H \rvert$.
    At terminal public states,
        the algorithm iterates (for each player) over all combinations of private infostates $(\privState_1, \dots, \privState_N) \in \prod_{\pl = 1}^N \privStates_\pl$.
        Since different private-infostate vectors are necessarily generated by different histories, we have
            $
                \lvert \bigcup_{\textnormal{terminal } \publicTerminalState} \prod_\pl \privStates_\pl(\publicTerminalState) \rvert
                \leq
                \lvert \mc Z \rvert
                \leq
                \lvert \mc H \rvert
                .
            $
    This shows that the time and memory complexity is at most $k' N \lvert \mc H \rvert$ for some small constant $k'$.
Moreover, this upper bound is necessarily tight because in perfect-information games, $\lvert \infostates_\pl \rvert = \lvert \mc H \rvert$.

(2) Going through the upper-bound calculation in (1) in more detail, we see that the complexity is bounded by
\begin{align*}
    &
    \sum_\pl
        \sum_{\publicState \in \publicTree}
            k
            \lvert
                \privStates_\pl(\publicState)
            \rvert
    +
    \sum_\pl
        \sum_{\publicTerminalState \textnormal{ terminal} }
            \lvert
                \prod_\pl \privStates_\pl(\publicTerminalState)
            \rvert
        \\
    & \leq
        kN \max_\pl \lvert \privStates_\pl \rvert
        + 
        N \prod_\pl \lvert \privStates_\pl \rvert
        \\
    & \leq
        2kN
        \prod_\pl
            \lvert \privStates_\pl \rvert
    .
\end{align*}

(3) The calculation from (2) shows that if
    we reduce the complexity of evaluating terminal public states
        from $\prod_\pl \lvert \privStates_\pl \rvert$
        to $C$,
    the overall complexity will reduce to
        $
            \bigO \left(
                \,
                \lvert \publicTree \rvert
                \cdot 
                \max \{ 
                 C , \,
                 \lvert \bigcup_\pl \privStates_\pl \rvert
                \}
                \,
            \right)
        $.
By \cite[Example 3]{acceleratedBR}, such more-efficient evaluation of terminal public states is possible
    in a number of domains,
    and two-player hold'em poker in particular admits $C = \sum_{\pl=1}^2 \lvert \privStates_\pl \rvert$.
(In poker, this requires a one-time investment to sort the private cards in each state based on their strength \cite{acceleratedBR}.
    However, we can hard-code this sorting or amortize its cost across all of the CFR iterations.)
This concludes the proof.
\end{proof}
\end{appendices}

\end{document}